\documentclass[aps,pra,twocolumn,showkeys]{revtex4-1}
\usepackage[english]{babel}
\usepackage{amsmath,paralist,amsthm,comment,amssymb}
\usepackage{times,color}
\usepackage{epsfig}
\usepackage{graphicx,float}
\usepackage[justification=raggedright]{caption}

\DeclareMathOperator{\supp}{supp}

\newcommand{\nix}[1]{}

\newtheorem{theorem}{Theorem}
\newtheorem{corollary}[theorem]{Corollary}
\newtheorem{lemma}[theorem]{Lemma}

\newtheorem{remark}{Remark}

\usepackage{newfloat,algcompatible}
\usepackage[size=small]{caption}
\usepackage{etoolbox}

\AtEndEnvironment{algorithm}{\noindent\hrulefill\par\nobreak\vskip-8pt}

\DeclareFloatingEnvironment[
    fileext=loa,
    listname=List of Algorithms,
    name=Algorithm,
    placement=tbhp,
]{algorithm}
\DeclareCaptionFormat{algorithms}{\vskip-15pt\hrulefill\par#1#2#3\vskip-6pt\hrulefill}
\algblock[Input]{Input}{EndInput}
\algblockdefx[Input]{Input}{EndInput}%
    [1]{\textbf{Input} #1}%
    {}

\begin{document}

\title{Projecting 3D color codes onto  3D toric codes}
\author{Arun B. Aloshious}
\email{aloshious.sp@gmail.com}
\author{Pradeep Kiran Sarvepalli}
\email{pradeep@ee.iitm.ac.in}
\affiliation{Department of Electrical Engineering, Indian Institute of Technology Madras, Chennai 600 036, India}
\date{\today}
\begin{abstract}
Toric codes and color codes are two important classes of topological codes.  Kubica, Yoshida, and Pastawski    showed that any $D$-dimensional color code can be mapped to a finite number of toric codes in $D$-dimensions. In this paper  we propose an alternate map of 3D color codes to 3D toric codes with a view to decoding 3D color codes. Our approach builds on Delfosse's result for 2D color codes and exploits the topological properties of these codes. Our result reduces the decoding of 3D color codes to that of 3D toric codes.  Bit flip errors are decoded by projecting on one set of 3D toric codes while phase flip errors are decoded by projecting onto another set of 3D toric codes. 
\end{abstract}
% insert suggested keywords - APS authors don't need to do this
%\keywords{topological codes, toric codes, color codes, decoding, 3D codes}

\maketitle
\section{Introduction}

Three dimensional (3D) toric codes \cite{castelnovo08,hamma05} and color codes \cite{bombin07} are topological quantum codes defined on 3D lattices. 
Like their 2D  counterparts, they are also Calderbank-Shor-Steane (CSS) codes \cite{calderbank96} where the bit flip and phase  errors can be corrected independently. 
Three dimensional topological codes are inherently asymmetrical  in their error correcting capabilities for the 
bit flip errors and the phase flip errors, which is perhaps one of the reasons why they have not received as much attention as their 2D analogues. 
However, the growing interest in asymmetric error models \cite{ioffe08,brooks13,sarvepalli09} motivates us to study them in closer detail. 

Another reason for studying 3D codes, specifically the 3D color codes, comes from the fact that in some ways they are also richer than 2D codes. Certain 3D color codes also possess a transversal non-Clifford gate \cite{bombin07}.  This is not possible for 2D  codes  or 3D toric codes. 
While the 3D toric code on a cubic lattice has been studied in \cite{castelnovo08,hamma05,bravyi11},  the general case of an arbitrary lattice has not been investigated as much. 

For these codes to be useful  for fault tolerant quantum computing it is necessary to develop efficient 
decoding algorithms. However, there appears to be no previous work on the decoding of 3D color codes. 
While we do not solve this problem in this paper, we make some progress in the decoding of 3D color codes by reducing it to the decoding of the 3D toric codes. Errors corresponding to chains in the lattice can be easily decoded on a 3D toric code. But efficient decoders are not known for errors corresponding to surfaces except in the case of cubic lattice. In this case a decoder similar to the decoder for the 4D toric code  in \cite{preskill02} can be used.

The central result of our paper is a mapping from 3D color codes to 3D toric codes.  By exploiting the topological properties of color codes, we establish a mapping between 3D color codes and 3D toric codes. 
The  work most similar to ours is that of Kubica, Yoshida, and Pastawski \cite{kubica15} who showed, among other things, that the 3D color code can be mapped  to three copies of 3D toric codes.  Our results give a different mapping from the color code to the toric codes. 
Our map also preserves the CSS nature of the color code. 
We project the $X$ errors and $Z$ errors onto different sets of toric codes unlike \cite{kubica15} which employs just one set of toric codes. 
Their map also implies that a 3D color code can be decoded via 3D toric codes. 
 The question of which map is better for decoding is not yet known. This will be investigated in a later work.

Another work related to ours is that of Delfosse  who showed that 2D color code can be projected onto surface codes \cite{delfosse14} using the machinery of chain complexes derived from hypergraphs. Our results generalize his approach to 3D. We take a somewhat simpler approach and do not explicitly make use of chain complexes based on hypergraphs. 
In passing we mention that similar mappings were known in 2D  \cite{yoshida11,bhagoji15,bombin11}. Their approaches  can also be generalized to  3D.

The paper is structured as follows. In Section~\ref{sec:bg} we give a brief review of  3D toric codes and color codes. 
In the subsequent section we present the central result of the paper showing how to project a 3D color code onto  a collection of 3D toric codes and propose a novel decoding scheme for color codes. We then conclude with a  brief discussion and outlook for further research. 
We assume that the reader is familiar with stabilizer codes \cite{calderbank98,gottesman97}. 

\section{Preliminaries}\label{sec:bg}
\subsection{3D toric codes}
We briefly review 3D topological codes.
A 3D toric code is defined over a cell complex (denoted $\Gamma$) in 3D. We assume that  qubits are placed on the edges of the complex. For each vertex $v$ and face $f$, 
we define  stabilizer generators as follows:
\begin{eqnarray}
A_v = \prod_{e \in \iota(v)} X_e \mbox{ and } B_f = \prod_{e \in \partial(f)} Z_e\label{eq:stabilizers-3d},
\end{eqnarray}
where $\iota(v)$ is the set of edges incident on $v$ and $\partial(f)$ is the set of edges that constitute the boundary of $f$. 
When there are periodic boundary conditions, the stabilizer generators $A_v$ and $B_f$  are constrained as 
follows:
\begin{eqnarray}
\prod_v A_v = I \text{ and }  \prod_{f \in \partial(\nu)} B_f= I,
\end{eqnarray}
where $\nu$ is any 3-cell and $\partial(\nu)$ is the collection of faces that form the boundary of $\nu$.
If $\Gamma$ has boundaries, then the constraints have to be modified accordingly. 
Additional constraints could be present depending on the cell complex. 

Sometimes it is useful to define the (3D) toric codes using the dual complex. 
We denote the dual of $\Gamma$ by $\Gamma^\ast$.
Qubits are placed on the faces of the dual complex. The stabilizer generators for a 3-cell $\nu$ and an edge $e$ in the dual complex are defined as
\begin{eqnarray}
A_{\nu} = \prod_{f\in \partial (\nu)} X_f
\mbox{ and } B_{e} = \prod_{f: e  \in \partial(f)} Z_{f}\label{eq:stabilizers-3d-dual},
\end{eqnarray} 
where $\partial(\nu)$ is the boundary of $\nu$ and $\partial (f) $ the boundary of $f$.

Phase flip errors on the toric code are detected by the operators $A_v$. 
Phase flip errors can be visualized as paths or strings on the lattice.
The nonzero syndromes always occur in pairs. 
The bit flip errors on the other hand are detected by the operators $B_e$. They are better visualized in the dual complex.  Since qubits are associated to faces in the dual complex, an $X$ error can be viewed as a surface obtained by union of faces (with errors)   and the (nonzero) syndrome as the boundary of the surface. Also note that since the boundary of each face is a cycle of trivial homology, the syndrome of $X$ errors is a collection of  cycles of trivial homology in $\Gamma^\ast$. 

\subsection{3D color codes}
Consider a complex with 4-valent vertices and 3-cells that are 4-colorable. Such colored complexes are called 3-colexes, \cite{bombin07}. 
A 3D color code is a topological stabilizer code constructed from a 3-colex. The stabilizer generators of the color code are given as
\begin{eqnarray}
B_\nu^X = \prod_{v\in \nu} X_v \mbox{ and }  B_f^Z = \prod_{v\in f} Z_v \label{eq:3d-tcc-stabilizers}
\end{eqnarray}
where  $\nu$  is a 3-cell and  $f$ a face. 
It turns out that for each 3-cell $\nu $ we can define a (dependent) $Z$ stabilizer as
$B^Z_\nu=\prod_{v\in \nu} Z_v$.
A 3-colex complex defines a stabilizer code with the parameters $[[v,3h_1]]$
where $h_1$ is the first Betti number of the complex, \cite{bombin07,bombin07a}. 

We can also define the color code in terms of the dual complex.  Now  qubits correspond to 3-cells, $X$-stabilizer generators to vertices and $Z$-stabilizer generators  to edges of $\Gamma^\ast$.
\begin{eqnarray}
B_v^X = \prod_{\nu: v\in  \nu} X_\nu \mbox{ and }  B_e^Z = \prod_{\nu: e\in  \nu} Z_\nu \label{eq:3d-tcc-stabilizers-dual}
\end{eqnarray}

We quickly review some relevant colorability properties of 3-colexes.
The edges of such a 3-colex can also be 4-colored: the outgoing edges of every 3-cell can be colored with the same color as the 3-cell.
We can  color the faces based on the colors of the 3-cells. A face is adjacent to exactly two 3-cells. A face adjacent to $3$-cells  colored $c$ and $c'$ is colored $cc'$. This means that the 3-colex is $6$-face-colorable. In view of the colorability of the 3-colex we refer to a $c$-colored  $3$-cell  as $c$-cell without explicitly mentioning that it is a 3-cell. Likewise we can unambiguously refer to the $cc'$-colored faces as
$cc'$-cells or $cc'$-faces,  $c$-colored edges as $c$-edges and $c$-colored vertices as $c$-vertices. 
We denote the $i$-dimensional cells of a complex $\Gamma$ as $\mathsf{C}_i(\Gamma)$ and  the $i$-dimensional cells of color $c$ as $\mathsf{C}_{i}^{c}(\Gamma)$. 

\section{Projecting a 3D color code onto 3D toric codes}\label{sec:embedding}

In this section we state and prove the central result of the paper, namely,   3D color codes can be projected onto a finite collection of  3D toric codes.  A more precise statement will be given later.  First, we give an intuitive explanation and then proceed to prove it rigorously.

\subsection{Intuitive  explanation through the decoding problem }

The main intuition behind the projection of color codes onto  toric codes is that any such mapping should preserve the error correcting 
capabilities  of the color code and enable decoding. From the point of view of a decoder, the information available to it is simply the syndrome information.
In a topological code this syndrome information can be represented by  
the cell complex. Our main goal is to preserve the syndrome information on the 3-colex while translating it into a different cell complex. 

In a 3-colex, qubits reside on the vertices while the checks correspond to faces and volumes. If we look at the dual complex,
the qubits correspond to 3-cells which are tetrahedrons; the $X$-type checks to vertices and the $Z$-type checks to 
edges. Due to this correspondence we often refer to the boundary of a qubit or a collection of qubits wherein  we mean the boundary of the 3-cells which correspond to those qubits in the dual complex.

We  address the bit flip and phase flip errors separately. Suppose that an $X$ error occurs.  
Since the qubits correspond to volumes, error correction is equivalent to (i) identifying
the boundary which {\em encloses} the qubits in error and (ii) specifying whether the erroneous qubits lie inside or outside the boundary. 
The second step is necessary because the qubits with errors have the same boundary as the qubits without errors.
 In case of bit flip errors, the syndrome information  is present on the edges; this is clearly not the boundary of a volume. The question then arises how do we recover the boundary of the erroneous qubits when we appear to be in possession of some partial information about the boundary. 

To see how we might solve this problem,  let us assume that there is just one bit flip error, see 
Fig.~\ref{fig:single-tetra-x} for illustration. This causes all the six edges of the tetrahedron to carry nonzero syndromes. While these edges are contained in the boundary of the tetrahedron it is  not the surface we are looking for. One way to recover the boundary of the tetrahedron is as follows. Imagine we deleted one vertex of the tetrahedron, then we would also be deleting three of the four faces of the tetrahedron and we would end with just one face. All the edges of this remaining face  carry nonzero syndromes. These edges are precisely the boundary of that face. Similarly deleting other vertices of the tetrahedron (separately) we would be able to recover all the faces of the tetrahedron. Since the union of these faces constitutes the boundary of the erroneous qubit we are able to recover the boundary of the error.
However, error correction is not complete. 
Both the single qubit in error and the collection of the qubits without error have the same boundary.
  To complete error correction we also need to choose which of these sets  of qubits are in error. We can decide on the volume  which contains fewer number of qubits. 
In the present case we would choose the qubit in error completing the error correction. 
Let us identify the key ideas in the previous procedure: 
\begin{compactenum}[(i)]
\item  We construct a  collection of complexes obtained by deleting $c$-vertices of tetrahedrons. 
\item  Then in each complex, from the edges carrying nonzero syndrome we recover  part of the boundary of the erroneous qubits. 
\item  Then we combine the boundary pieces found in (ii) to recover the boundary of the erroneous qubits. 
\item  Finally we decide whether the interior or the exterior set of qubits enclosed by the boundary are in error. 
\end{compactenum}
Step (ii) is key to making the connection with the 3D toric codes. This step is identical to the correction of the $X$-type errors in 3D toric codes.

\begin{figure}
%\centering
\includegraphics[scale=0.4]{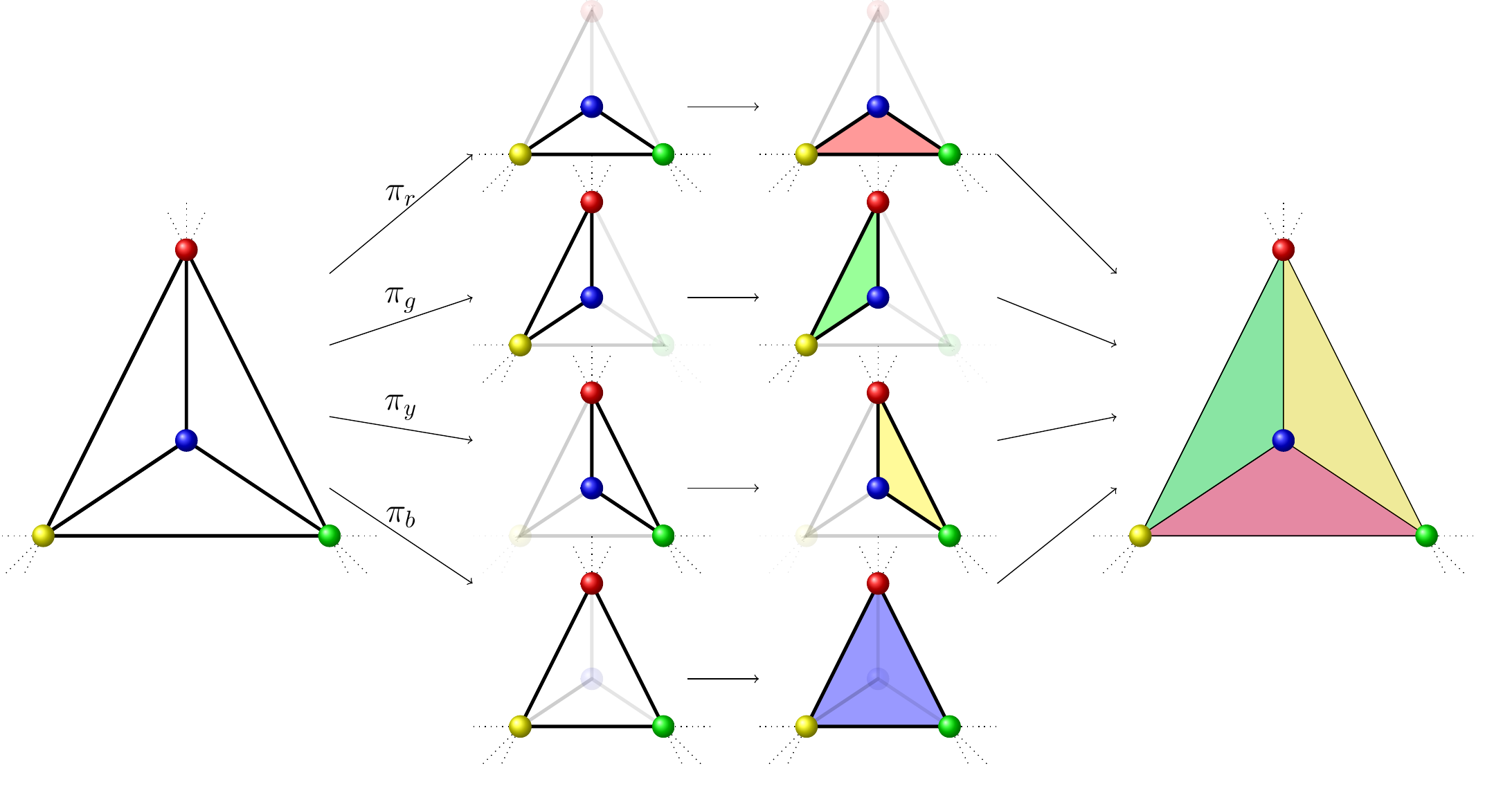}
\caption{We can recover the boundary of a tetrahedron from the edges by reconstructing the faces and then combining the faces.}
\label{fig:single-tetra-x}
\end{figure}

A similar idea will lead us to the procedure for decoding the phase flip errors, see Fig.~\ref{fig:single-tetra-z}. In this case the syndromes that detect the phase flip errors correspond to the vertices in the dual complex. Suppose now that there is a single phase flip error. The vertices of the erroneous tetrahedron will carry the syndrome information about the error. Now we seem to have even lesser information about the boundary of the erroneous tetrahedron than before. However, we can recover the boundary by the following procedure. Delete any pair of vertices of the tetrahedron. We will be left with one edge and two vertices. We can first identify the edge as piece of the boundary we are looking for. 
Deleting  all the six possible pairs of vertices, we are able to recover the six edges which are in the boundary of the tetrahedron. Now the problem is identical to the one we solved for correcting bit flip errors. 
We summarize the key steps:
\begin{compactenum}[(i)]
\item We construct new complexes from the original complex by deleting pairs of vertices of each tetrahedron. 
\item Then we recover the edges which are in the boundary of the erroneous qubits. 
\item At this point the problem is same as the problem of decoding bit flip errors which can be solved using the previous procedure. 
\end{compactenum}
In correcting the $Z$-type errors the connection to the toric codes happens in (ii). This is precisely the process used to decode 
$Z$-type errors in 3D toric codes.

The procedures we outlined are heuristic and somewhat imprecise; 
they need a  rigorous justification as to correctness and efficiency. 
We also need to consider the cases where the boundary recovery procedure can fail. 
We now turn to address these issues in the next section.

\begin{figure}
%\centering
\includegraphics[scale=0.4]{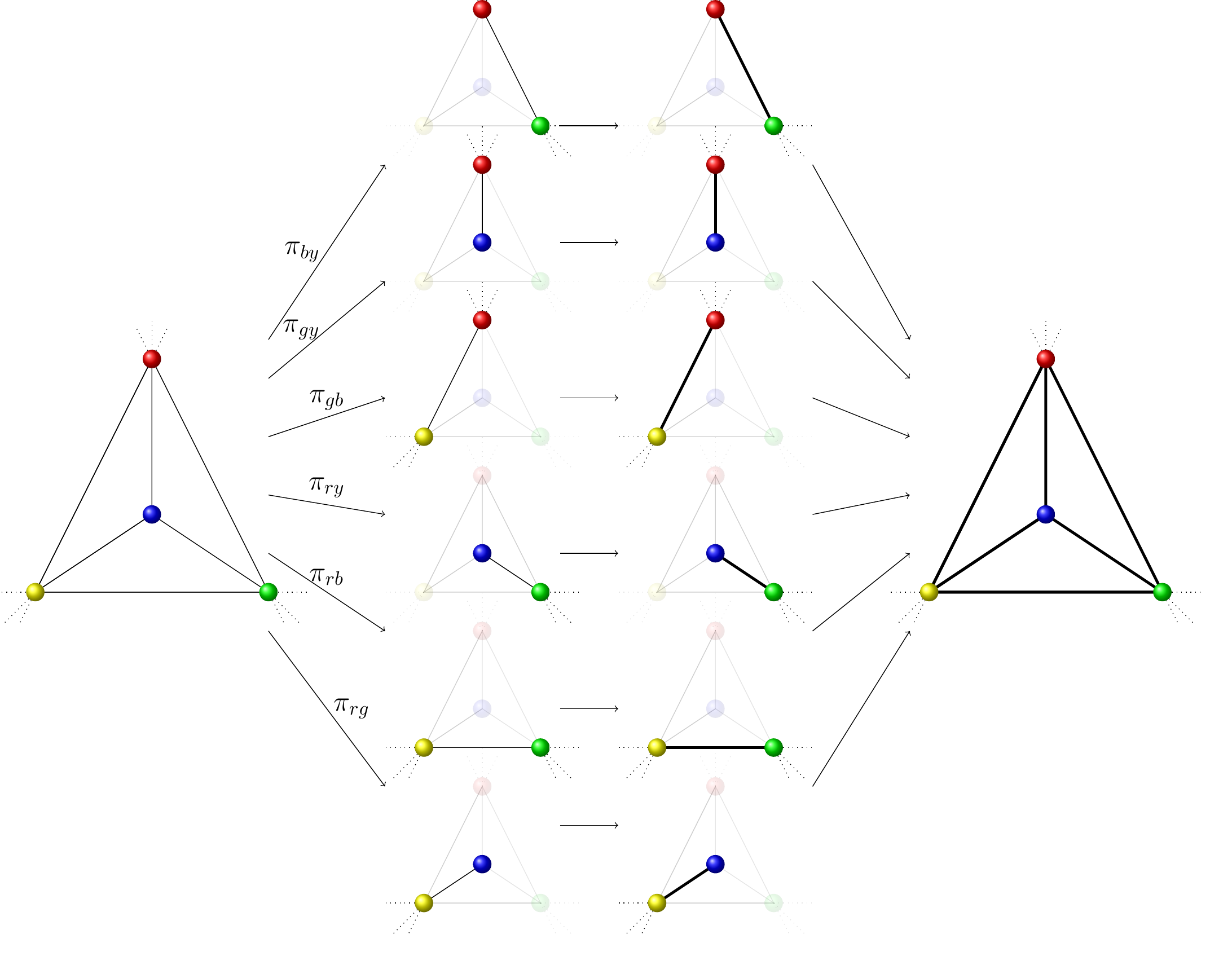}
\caption{We can recover the boundary of a tetrahedron from the vertices by first recovering the edges in the boundary of the tetrahedron. With the edges recovered we can proceed as illustrated in Fig.\ref{fig:single-tetra-x} to recover the boundary of the tetrahedron.}
\label{fig:single-tetra-z}
\end{figure}

\subsection{3-colexes, duals and minors}

As we saw in the previous section, our approach to decoding color codes leads us to duals and minors of complexes. So we begin by studying the properties of the 3-colexes and their minors. 
First we state some properties of the dual of a 3-colex.  Since they are immediate from the properties of the 3-colex  we omit the proof. 
For the rest of the paper we assume that $\Gamma$ is a 3-colex.

\begin{lemma}\label{lm:face-colorability-dual-colex}
Let $\Gamma$ be a 3-colex. Then the dual complex $\Gamma^\ast$ is 4-vertex-colorable, 6-edge-colorable,
4-face-colorable. 
\end{lemma}

Every qubit in the 3-colex  is identified with a tetravalent vertex incident on four 3-cells which are 4-colorable. Therefore in the dual complex, every qubit corresponds to a tetrahedron whose vertices are of different colors. Similarly the four faces of each tetrahedron are also of different colors. (Follows from  Lemma~\ref{lm:face-colorability-dual-colex}.)
Let us denote the minor complex of  $\Gamma^\ast$, i.e. the complex obtained by deleting all the vertices of color $c$, by $\Gamma^{\ast \setminus c}$. 
We denote this operation as $\pi_c$ so that 
\begin{eqnarray}
\pi_c(\Gamma^\ast) =  \Gamma^{\ast\setminus c} \label{eq:minor-map-1}
\end{eqnarray}
The resulting structure $\Gamma^{\ast\setminus c}$ is a well defined complex. Clearly, its vertices and edges are a subset of the parent complex $\Gamma^\ast$. The structure of the faces and 3-cells is not so obvious. 
A face in $\Gamma^\ast$ that is incident on a $c$-vertex will not survive in $\Gamma^{\ast\setminus c}$. Therefore only the $c$-faces of  $\Gamma^\ast$ which are not incident on a $c$-vertex will be faces of 
$\Gamma^{\ast\setminus c}$. 
The 3-cells of $\Gamma^{\ast\setminus c}$ are formed by merging all the tetrahedrons that are incident on a 
$c$-vertex. 

We can now extend the action of $\pi_c$ to the individual cells of $\Gamma^\ast$, there is some freedom on how to extend as long as we retain the information needed for error correction. 
We are primarily  interested in extending  $\pi_c$ so that it captures the information about (i) the qubits, (ii) the errors on them and (iii) the associated syndrome. 

A vertex that is not colored $c$ will be mapped to a vertex in $\Gamma^{\ast\setminus c}$.
An edge that is not incident on a  $c$-vertex will be mapped to an edge in $\Gamma^{\ast\setminus c}$. Edges that are incident on a $c$-vertex can be thought as being mapped to the empty set. 
A $c$-face in $\Gamma^{\ast}$  will continue to be a face in  $\Gamma^{\ast\setminus c}$
as it is not incident on any $c$-vertices. 
A 3-cell has exactly 4 faces and on the deletion of a $c$-vertex just the $c$-face in its boundary will be left in $\Gamma^{\ast\setminus c}$. A $3$-cell  in $\Gamma^\ast$ corresponds to a qubit, so we can interpret the 
$c$-face in its boundary as the qubit in $\Gamma^{\ast\setminus c}$. Since every face is shared between two 3-cells, there exist two 
distinct  cells $\nu_1$ and $\nu_2$ such that $\pi_c(\nu_1) = \pi_c(\nu_2)$.
The following equations summarize the preceding discussion. 
\begin{subequations}
\begin{eqnarray}
\pi_c(v)&=&v \mbox{ if  } v\in \mathsf{C}_0(\Gamma^\ast) \setminus \mathsf{C}_0^c(\Gamma^\ast) \\
\pi_c(e)&=&e \mbox{ if  } e  \in \mathsf{C}_1(\Gamma^\ast)\setminus \mathsf{C}_1^{cc'}(\Gamma^\ast)\label{eq:pi-edge} \\
\pi_c(f) & = &f \mbox { if } f\in \mathsf{C}_2^c(\Gamma^\ast)\label{eq:pi-face}\\
\pi_c(\nu)&=&f_{\ni\nu}^c=\partial \nu \cap \mathsf{C}_2^{c}(\Gamma^\ast); \nu \in \mathsf{C}_3(\Gamma^\ast)\label{eq:cell-bndry}
\end{eqnarray}
\end{subequations}
where $f_{\ni\nu}^c$ is the unique $c$-face in $\nu$. From these relations we can write the boundary of a 3-cell as 
\begin{eqnarray}
\partial(\nu) = \sum_c \pi_c(\nu) = \sum_c f_{\ni\nu}^c, \label{eq:qubit-bndry}
\end{eqnarray}
where the summation is carried  addition modulo 2. 
The boundary of a collection of 3-cells can be extended linearly. 

Some relevant properties of the minor complexes are considered next. They concern both the structure and coloring properties of the  minor complexes. 
\begin{lemma}\label{lm:monochromatic-faces-minor-complex}
Let $\Gamma$ be a 3-colex and $ c\in \{r,b,g,y \}$.
Then the minor complex $\Gamma^{\ast \setminus c} $
has only  $c$-colored faces and $dd'$-edges  where $d,d'\in \{r,b,g,y \}\setminus \{c\}$. 
\end{lemma}
\begin{proof}
Suppose that we delete the vertices colored $c$ in $\Gamma^{\ast\setminus c}$; this leads to the deletion of the edges and faces that are incident on these vertices. On any $c$-colored vertex only the $c$-colored faces are not incident. 
Any face colored with $c'\in \{ r,b,g,y\}\setminus c$ is incident on some $c$-vertex. This leads to deletion of all but one face of each tetrahedron incident on any $c$-colored vertex. The remaining face is colored $c$. Thus $\Gamma^{\ast\setminus c}$
contains only  $c$-colored faces. Since only $d,d'$-vertices are present in $\Gamma^{\ast\setminus c}$, the edges connecting them are colored $dd'$. 
\end{proof}

\begin{lemma}\label{lm:3cells-minor}
The 3-cells of $\Gamma^{\ast\setminus c}$ can be indexed by the $c$-vertices of $\Gamma^\ast$ and 
the boundary of a 3-cell $\nu_v \in \mathsf{C}_3(\Gamma^{\ast\setminus c})$ is the  sum of $c$-faces of  tetrahedrons incident on $v$.
\end{lemma}
\begin{proof}
The deletion of  a $c$-vertex causes all the tetrahedrons incident on it to be combined into one single 3-cell in $\Gamma^{\ast\setminus c}$. 
Since the four vertices of each tetrahedron are different colors, two tetrahedrons can be merged only if they are incident on the same $c$-vertex. Thus two distinct $c$-vertices lead to distinct 3-cells. Hence, the 3-cells of $\Gamma^{\ast\setminus c}$ can be indexed by the $c$-vertices of $\Gamma^\ast$.
The deletion of the $c$-vertex $v$, creates a 3-cell and leaves behind a $c$-face for every tetrahedron incident on $v$. 
These $c$-faces enclose the 3-cell formed by merging the qubits incident on $v$, therefore they must form its boundary.
Denote by $\nu_v$ such a 3-cell. Then its boundary $\partial(\nu_v)$ is given by 
\begin{eqnarray}
\partial(\nu_v) & =  & \sum_{\nu:v\in \nu}\partial(\nu) \overset{(a)}{=}\sum_{\nu:v\in \nu} f_{\ni\nu}^c +\sum_{\nu:v\in \nu} \sum_{i\neq  c } f_{\ni\nu}^i\\
&\overset{(b)}{=}&\sum_{\nu:v\in \nu} f_{\ni\nu}^c,
\end{eqnarray}
which is precisely the sum of $c$-faces of tetrahedrons incident on $v$. Note that 
 $(a)$ follows from Eq.~\eqref{eq:qubit-bndry} while  $(b)$ is due to the fact that every $c'$-face incident on $v$ is shared between exactly two qubits incident 
on $v$  causing the second summation to vanish.
\end{proof}

Let $\Gamma^{\ast \setminus cc'}$ denote  the minor of  $\Gamma^\ast$ obtained by deleting all vertices colored $c$ and $c'$.
We assume that $c\neq c'$ for the rest of the paper.  Denote this operation as $\pi_{cc'}$. Then we have 
\begin{eqnarray}
\pi_{cc'}(\Gamma^\ast)= \Gamma^{\ast\setminus cc'} \label{eq:minor-map-2}.
\end{eqnarray}
Note that the order of deletion of vertices does not matter, therefore we have 
\begin{eqnarray}
\pi_{cc'}(\Gamma^\ast)=  \pi_{c'c}(\Gamma^\ast)\label{eq:minor-map-symmetry}.
\end{eqnarray}

As in  case of $\Gamma^{\ast \setminus c}$, the vertices and edges of of $\Gamma^{\ast \setminus cc'}$ are a subset of $\Gamma^{\ast \setminus cc'}$ and can be easily identified. The faces and 3-cells are not so obvious. 
The faces of $\Gamma^{\ast\setminus cc'} $ are not faces in the parent complexes $\Gamma^\ast$  or  $\Gamma^{\ast\setminus c}$. We need to define the faces and 3-cells of $\Gamma^{\ast\setminus cc'} $. We make one small observation before defining them. 
\begin{lemma}\label{lm:bichromatic-edges-minor-complex}
Let $\Gamma$ be a 3-colex and $ c,c'\in \{r,b,g,y \}$. Then  $\Gamma^{\ast\setminus cc'}$ has only $dd'$-colored edges 
where $\{d,d'\} =\{ r,b,g,y\}\setminus \{c,c' \}$. 
\end{lemma}
\begin{proof}
Suppose that we delete all the vertices colored $c,c'$ in $\Gamma^\ast$, then all edges incident on $c$-vertices
and $c'$-vertices will be deleted. Thus only edges between $d$ and $d'$ colored vertices will remain. These
edges are colored $dd'$. 
\end{proof}

Let $e$ be a $cc'$-edge $e$ in  $\Gamma^\ast$ and  $\mathcal{V}_e$ the qubits containing  $e$. 
\begin{eqnarray}
\mathcal{V}_e=\{ \nu \mid e\in \nu \}
\end{eqnarray}
Let  $e_{\ni \nu}^{dd'}$ be the unique $dd'$-edge in $\nu$.
For each of the qubits in $\mathcal{V}_e$, exactly one $dd'$-edge will survive in  $\Gamma^{\ast\setminus cc'} $. 
No two qubits in $\mathcal{V}_e$ share the same $dd'$-edge. 
Further, the surviving $dd'$ edges will form a 
cycle in $\Gamma^{\ast\setminus cc'}$,  see Fig.~\ref{fig:face-cc-complex} for illustration. 
This can be seen as follows. 
Let  $\mathsf{V}_{cc'}( \mathcal{V}_e)$,  denote the vertices of $\mathcal{V}_e$  that remain in $ \Gamma^{\ast\setminus cc'}$. 
Every such vertex $v$ is incident on exactly two qubits of $\mathcal{V}_e$. Therefore, two $dd'$-edges are incident on $v$.
Hence, the $dd'$-edges of $\mathcal{V}_e$ form  a cycle.
This cycle is of trivial homology since it is on the boundary of a 3-cell (formed by the qubits in $\mathcal{V}_e$). We can associate a face to this cycle such that it lies entirely in the 3-cell. In other words,   to every $cc'$-edge in $\Gamma^\ast$, we can associate a face in $\Gamma^{\ast\setminus cc'}$. The boundary of this face is the collection of the $dd'$-edges belonging to the qubits incident on $e$, alternatively, 
\begin{eqnarray}
\partial f_e & = &\sum_{\nu \in \mathcal{V}_e } e_{\ni \nu}^{dd'}. \label{eq:bndry-f-cc}
\end{eqnarray} 
\begin{figure}[ht]
\includegraphics[scale=0.6]{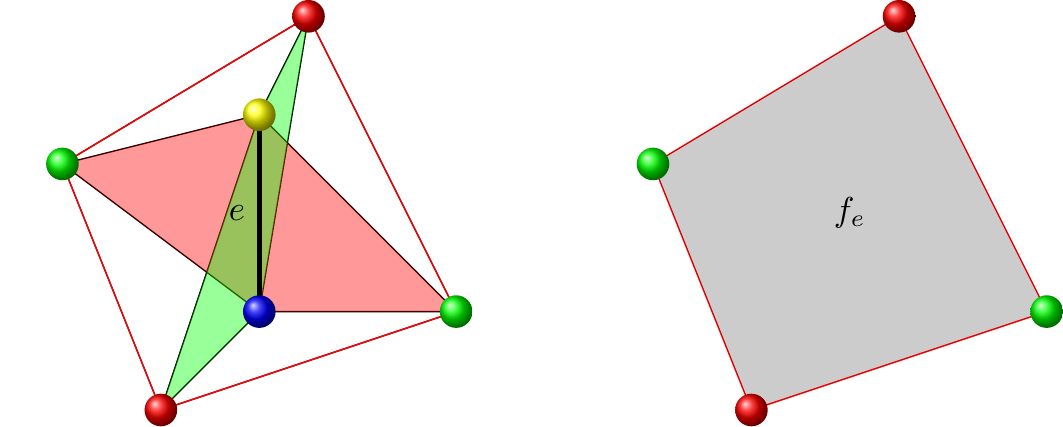}
\caption{Face in $\Gamma^{\ast\setminus cc'}$. Consider the $by$-edge (in bold) and the qubits incident on it. Since only  $r$,  $g$-vertices (i.e. red and green)  will survive in 
$\Gamma^{\ast\setminus by}$, each of these surviving vertices will have exactly two $rg$-edges incident on them. Therefore these $rg$-edges form a cycle.}
\label{fig:face-cc-complex}
\end{figure}

The qubits incident on  distinct $cc'$-edges $e_1$ and $e_2$ will be disjoint intersecting in either edges or vertices, so the faces associated to them, i.e. $f_{e_1}$ and  $f_{e_2}$, will also be disjoint and intersect in edges or vertices. Therefore the faces 
are well defined. 
The preceding discussion proves the following result. 
\begin{lemma}\label{lm:faces-in-cc}
Let $\Gamma$ be a 3-colex and $ c,c'\in \{r,b,g,y \}$. Then the faces of the minor complex $\Gamma^{\ast\setminus cc'}$   are in one to one correspondence with $cc'$ edges of $\Gamma^*$. 
\end{lemma}

With faces of $\Gamma^{\ast\setminus cc'}$ defined, the 3-cells of $\Gamma^{\ast\setminus cc'}$ can be identified. Let $v$ be a $c$ or $c'$-vertex. 
Consider the  edges incident on $v$. 
Each of these edges corresponds to a face in $\Gamma^{\ast\setminus cc'}$, by Lemma~\ref{lm:faces-in-cc}. 
We define the volume enclosed by these faces, such that it contains $v$, to be a 3-cell of $\Gamma^{\ast\setminus cc'}$. 
Since every such $v$ leads to a 3-cell in $\Gamma^{\ast\setminus cc'}$, we have the following. 

\begin{lemma}\label{lm:3cells-in-cc}
Let $\Gamma$ be a 3-colex and $ c,c'\in \{r,b,g,y \}$. Then 3-cells of the minor complex $\Gamma^{\ast\setminus cc'}$  are in one to one correspondence with vertices in $\mathsf{C}_0^c(\Gamma^*)\cup \mathsf{C}_0^{c'}(\Gamma^*)$. 
\end{lemma}

We can also extend $\pi_{cc'}$ to the cells of $\Gamma^\ast$ as we did for $\pi_c$. Then we can write 
\begin{subequations}
\begin{eqnarray}
\pi_{cc'}(v)&=&v \mbox{ if  } v\in \mathsf{C}_0(\Gamma^\ast) \setminus (\mathsf{C}_0^c(\Gamma^\ast) \cup \mathsf{C}_0^{c'}(\Gamma^\ast)) \label{eq:pi-cc-site}\\
\pi_{cc'}(e)&=&e \mbox{ if  } e  \in \mathsf{C}_1^{dd'}(\Gamma^\ast); d,d'\not\in \{c,c' \} \label{eq:pi-cc-edge}\\
\pi_{cc'}(\nu)&=&e_{\ni \nu}^{dd'}=\partial(\pi_c(\nu)) \cap \mathsf{C}_1^{dd'}(\Gamma^\ast); \nu \in \mathsf{C}_3(\Gamma^\ast) \label{eq:pi-cc-cell}
\end{eqnarray}
\end{subequations}
where $e_{\in \nu}^{dd'}$ is the unique $dd'$-edge in $\nu$.
In these equations  and henceforth we assume $d,d'\in \{r,b,g,y \}\setminus \{c,c'\} $ and $d\neq d'$. None of the faces of $\Gamma^{\ast\setminus c}$ or $\Gamma^\ast$  will survive in 
$\Gamma^{\ast\setminus cc'} $. 
Since the faces of $\Gamma^{\ast }$ do not carry any information about the qubits and the $Z$ error syndromes, we are not particularly interested in them; we have some freedom  as to how to define $\pi_{cc'}$ for faces in $\Gamma^{\ast\setminus c}$.

We define the edge boundary of a 3-cell in $\Gamma^{\ast}$ as 
\begin{eqnarray}
\delta \nu = \sum_{cc'}\pi_{cc'}(\nu)= \sum_{cc'} e_{\ni \nu}^{dd'}
 \label{eq:edge-bndry-qubit}
\end{eqnarray}
We call it the edge boundary because $\pi_{cc'}(\nu)$ is an edge, see Eq.~\eqref{eq:pi-cc-cell}.
We extend $\delta$ to multiple 3-cells  linearly.

A simple example of 3-colex and related complexes are shown in Fig.~\ref{fig:color_code}. We can relate the various complexes and the objects of interest for us as follows. 
\begin{center}
\begin{tabular}{l|c|c|c|c}\hline
&$\Gamma$ &  $\Gamma^\ast$   &$\Gamma^{\ast\setminus c}$&$\Gamma^{\ast\setminus cc'}$  \\
        \hline
        Qubit& vertex & tetrahedron & triangle & edge \\
        $Z$-check & face &  edge & edge & ---\\
        $X$-check  & 3-cell & vertex & vertex & vertex\\ \hline
\end{tabular}
\end{center}
\begin{figure*}[t]
\centering
\includegraphics{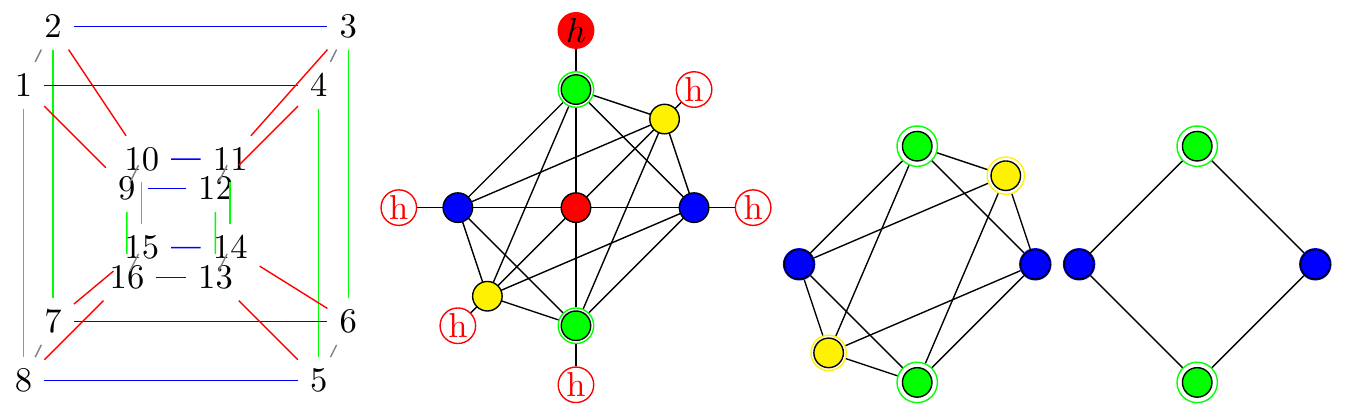}
    \caption{A 3-colex $\Gamma$ and its dual $\Gamma^*$; $i$-cells of $\Gamma^\ast$ correspond to the $3-i$ cells of $\Gamma$. The minor complexes $\Gamma^{* \setminus r}$ and $\Gamma^{* \setminus ry}$ are also shown. }
    \label{fig:color_code}
\end{figure*}

The preceding lemmas lead to the following corollary. 
\begin{corollary}\label{co:complexity-of-maps}
Let $\Gamma$ be a 3-colex with $v$ vertices, $e=2v$ edges,  $f_{cc'}$ $cc'$-faces, $\nu_c$ $c$-cells.
Let the total number of faces be $f=\sum_{cc'} f_{cc'}$ and 3-cells be $\nu=\sum_c \nu_{c}$.
The following table summarizes the number of cells in $\Gamma^{\ast\setminus c}$ and $\Gamma^{\ast\setminus cc'}$ where 
$c,c',d,d'\in\{ r,b,g,y\}$ are distinct. 
\end{corollary}
\begin{center}
\begin{tabular}{l|c|c|c|c}\hline
&$\Gamma$ &  $\Gamma^\ast$   &$\Gamma^{\ast\setminus c}$&$\Gamma^{\ast\setminus cc'}$  \\
        \hline
        3-cells&$\nu$ & $v$ &$\nu_c$ & $\nu_{c}+\nu_{c'}$ \\ 
        Faces & $f$ & $2v$&$ v/2$ & $f_{cc'}$\\ 
        Edges & $2v$ &$f$ &$ f_{dc'}+f_{c'd'}+f_{dd'}$ & $f_{dd'}$\\
        Vertices & $v$ & $\nu$&$ \nu_{c'}+\nu_{d}+\nu_{d'}$ &$\nu_{d}+\nu_{d'}$\\ \hline
\end{tabular}
\end{center}

\begin{proof}
The $i$-cells in dual complex $\Gamma^\ast$ are in one to one correspondence with the $(3-i)$-cells of $\Gamma$. Now suppose that $\Gamma^\ast$ is modified so that all vertices colored $i$ are deleted.
 Then all 3-cells incident on it will be merged to form a new 3-cell. Since all the 3-cells in $\Gamma^{\ast}$ are incident on some 
  $i$-vertex, they will be part of some new 3-cell and   $\Gamma^{\ast\setminus c}$ will contain $c_i$ 3-cells.
 On deleting the $i$-vertices, exactly one face will remain from each 3-cell in $\Gamma^\ast$. Since each of these faces will be shared between two 3-cells there will be $v/2$ faces. The edges in $\Gamma^{\ast \setminus i}$ are those that are incident on vertices other than 
 $i$-vertices.
 
Similarly in $\Gamma^{\ast\setminus cc'}$, only the $d$-vertices,  $d'$-vertices  and the $dd'$-edges will survive.
These are precisely $\nu_d+\nu_{d'}$ vertices and $f_{dd'}$ edges.
The number of  faces and 3-cells of $\Gamma^{\ast\setminus cc'}$ is immediate from 
Lemmas~\ref{lm:faces-in-cc}~and~\ref{lm:3cells-in-cc}.
\end{proof}

\begin{remark}
The minor complexes defined here are the duals of the shrunk complexes defined in \cite{bombin07a}. For example, $\Gamma^{*\setminus b}$ is the exactly the dual of $b$-shrunk complex and $\Gamma^{*\setminus ry}$ is the dual of 
$ry$-shrunk complex.
\end{remark}

\subsection{$X$ type errors on 3D color codes}

Let us now see how to perform error correction on  a color code.  It is helpful to see the (topological) structure of the errors in the dual of the 
3-colex. 
We  analyze the bit flip and phase flip errors separately. Suppose that we have $X$ errors on some set of qubits. In the dual colex the erroneous qubits correspond to a volume. 
Through $\pi_c$ we can associate qubits to faces of the minor complex $\Gamma^{\ast\setminus c}$. Thus we can 
project errors from $\Gamma^\ast$ to $\Gamma^{\ast\setminus c}$.

If a qubit $\nu$ has a bit flip   error, then we place an $X$-error on the image of  $\nu$ in  $\Gamma^{\ast\setminus c}$.  In other words,
\begin{eqnarray}
\pi_c(X_\nu) = X_{\pi_c(\nu)},\label{eq:X-map-1}
\end{eqnarray}
where $\pi_c(X_\nu)$ gives an
$X$ error acting on the qubits in $\Gamma^{\ast\setminus c}$.

A consequence of Eq.~\eqref{eq:X-map-1}, together with linearity of $\pi_c$, is that for two adjacent qubits $\nu_1$ and 
$\nu_2$ sharing a $c$-colored face $f$, we have $\pi_c(X_{\nu_1}X_{\nu_2}) = X_{\pi_c(\nu_1)}X_{\pi_c(\nu_2)} = I$, where we
used the fact that $\pi_c(\nu_1)=\pi_c(\nu_2)=f$.
In other words, a $c$-face common to two qubits in error corresponds to an error free qubit in $\Gamma^{\ast\setminus c}$.

 The syndrome corresponding to bit flip errors is associated to edges of $\Gamma^\ast$. In $\Gamma^{\ast\setminus c}$ not all edges
 are present. But if an edge is present, we associate to that edge the same syndrome as in $\Gamma^\ast$. Syndromes in the minor complex are essentially the restriction of the syndromes in $\Gamma^\ast$. Let $s_e$ be the syndrome on edge $e$, then 
 \begin{eqnarray}
\pi_c(s_e)= s_{\pi_c(e)}  = s_e. \label{eq:x-syndrome-map}
 \end{eqnarray}
This is consistent with Eq.~\eqref{eq:pi-edge}.

 At this point we have qubits living on the faces and syndromes on the edges of $\Gamma^{\ast\setminus c}$ just as we would have in a 3D 
 toric code. But it needs to be shown that indeed that we truly have the structure of a 3D toric code  and not merely the appearance of 
 it. This we shall take up next.

 First let us consider the edge type checks on $\Gamma^{\ast\setminus c}$.
Consider an edge $e$ in $\Gamma^{\ast\setminus c}$. Then $e$ is also present in $\Gamma^\ast$. For each qubit incident on $e$ there is an $c$-colored face incident on $e$. These $c$-faces exhaust all the faces incident on $e$ in $\Gamma^{\ast\setminus c}$. Thus in 
the 3D toric code associated to $\Gamma^{\ast\setminus c}$, every face incident on  $e$ participates in that check on $e$ as required for the edge type checks in the 3D toric code. 
Next we look at the projected syndromes on the minor complex.

\begin{theorem}[Projection of $X$ errors onto toric codes]\label{lm:x-err-syn}
Let $s$ be the syndrome for an $X$ error $E$ on the 3D color code defined on a 3-colex $\Gamma$ and $\pi_c(s)$ the restriction of $s$
 on $\Gamma^{\ast\setminus c}$.
Then the error $\pi_c(E)$ in $\Gamma^{\ast\setminus c}$ produces the syndrome $\pi_c(s)$ in the toric code
associated to $\Gamma^{\ast\setminus c}$. 
\end{theorem} 
\begin{proof} 
We now will show that in 
$\Gamma^{\ast\setminus c}$ the syndrome produced by $\pi_c(E)$ is same as $\pi_c(s)$. Consider any edge in $\Gamma^{\ast\setminus c}$;
by definition  $\pi_c(s_e)=s_e$ where $s_e$ is the syndrome on $e$ with respect to $\Gamma^\ast$. Since $\Gamma$ is 3-colex, an even number of qubits are incident on $e$, say $2m$. Then $s_e=\oplus_{i=1}^{2m} q_i$ where $q_i=1$ if there is an $X$ error on the $i$th qubit and zero otherwise.
Each of these qubits (tetrahedrons) incident on $e$ are projected to a  qubit in $\Gamma^{\ast\setminus c}$. But note that two qubits which share a face are mapped to the same qubit in $\Gamma^{\ast\setminus c}$. Thus there are $m$ qubits (triangles) incident on $e$
with respect to $\Gamma^{\ast\setminus c}$.
These (projected) qubits are in error if and only if one of the parent qubits in $\Gamma^\ast$ are in error. Let  $r_j=1$ if there is an an error on the projected qubit and zero otherwise. Then $r_j=q_{2j-1}\oplus q_{2j}$, where $2j-1$ and $2j$ are the qubits which are projected onto the $j$the qubit in 
$\Gamma^{\ast\setminus c}$.
The syndrome on  $e$ as computed in the 3D toric code is $\oplus_{j=1}^{m} r_j  = 
\oplus_{j=1}^m(q_{2j-1}\oplus q_{2j}) = s_e$. Thus the projected error $\pi_c(E)$ produces the same syndrome as the projected syndrome 
$\pi_c(s)$.
\end{proof}

\begin{corollary}\label{co:valid-syndrome}
Let $s$ be the syndrome for an $X$ error  on the 3D color code defined on a 3-colex $\Gamma$ and $\pi_c(s)$ the restriction of $s$
 on $\Gamma^{\ast\setminus c}$. Then $\pi_c(s)$ is a valid syndrome for (an $X$ error on) the toric code on $\Gamma^{\ast\setminus c}$.
\end{corollary}
\begin{proof}
By Theorem~\ref{lm:x-err-syn}, $\pi_c(s)$ is the same as the syndrome produced by an $X$ error on $\Gamma^{\ast\setminus c}$. 
Therefore, it must be a valid syndrome for an $X$ error for the 3D toric code on $\Gamma^{\ast\setminus c}$.
\end{proof}

\begin{lemma}\label{lm:stabgen2stabgen}
Let $v$ be a c-vertex in $\Gamma^\ast$ and $\nu_v$ be the 3-cell in $\Gamma^{\ast\setminus c}$ obtained by merging all the qubits incident on $v$. 
Then the $X$-type stabilizer $B_v^X$ of the color code on 
$\Gamma^\ast$ is mapped to an  $X$-type stabilizer generator of the toric code on $\Gamma^{\ast\setminus c}$.  
\begin{eqnarray}
\pi_c(B_v^X) =  B_{\nu_v}^X \mbox{ and }\pi_{c'}(B_v^X) = I \mbox{ for }  c'\neq c
\end{eqnarray}
\end{lemma}
\begin{proof}
We have $B_v^X = \prod_{\nu :v\in \nu} X_\nu$. 
Then  
\begin{eqnarray}
\pi_c(B_v^X) & = & \pi_c\left(  \prod_{\nu:v\in \nu } X_\nu\right)  = \prod_{\nu:v\in \nu} X_{\pi_c(\nu)}\\
&\overset{(a)}{=}&\prod_{f\in \partial(\nu_v) }X_f  =B_{\nu_v}^X.
\end{eqnarray}
where $(a)$ follows from Lemma~\ref{lm:3cells-minor}. 
Thus $\pi_c(B_v^X) $  is exactly the $X$-type stabilizer generator defined on the $3$-cell $\nu_v$. 

Now consider a $c'$-face incident on $v$. Such a face is in the boundary of  two qubits incident on $v$.  
This means that for every qubit $\nu$ incident on $v$, there exists another qubit  $\nu'$ incident on $v$ such that $\pi_{c'}(X_\nu) = \pi_{c'}(X_{\nu'})$.
Therefore $\pi_{c'}(B_v^X) =\prod_{\nu:v\in \nu} X_{\pi_c(\nu)}=I$.
\end{proof}
Note that the previous lemma implies that for a $c'$-vertex $\pi_c(B_v^X)=I$. 
\begin{corollary}\label{co:stab2stab}
Let $\overline{S}$ be an $X$-type stabilizer on the color code. Then $\pi_c(\overline{S})$ is an $X$-stabilizer on the toric code 
defined by $\Gamma^{\ast\setminus c}$. Conversely for every $X$-stabilizer $S$ on $\Gamma^{\ast\setminus c}$, there exists an $X$-stabilizer generator 
$\overline{S}$ on $\Gamma^\ast$ such that $\pi_c(\overline{S}) = S$ and $\pi_{c'}(\overline{S})=I$.
\end{corollary}
\begin{proof}
The first statement is a consequence of Lemma~\ref{lm:stabgen2stabgen} as $\{B_v^X\}$ generate all $X$-stabilizers of the color code. 
Further, $\{ \pi_c(B_v^X) \} =\{ B_{\nu_v}^X\}$, where $v\in \mathsf{C}_0^c(\Gamma^\ast)$. By Lemma~\ref{lm:3cells-minor},  $\{ B_{\nu_v}^X\}$ generate the $X$-type stabilizers of the toric code on 
$\Gamma^{\ast\setminus c}$. Thus the  converse also holds.
\end{proof}

Let the support of an error $E$ be defined\footnote{Usually the support is defined as a set, here it is convenient to define as a linear combination.} as
\begin{eqnarray}
\supp(E) = \sum_{i:E_i\neq I} i\label{eq:support},
\end{eqnarray}
where $i$ could be a 3-cell, face, or an edge depending on where the qubits are located. 
It follows that 
\begin{eqnarray}
\supp(EE') = \supp(E) +\supp(E').\label{eq:supp-sum}
\end{eqnarray}
We define the boundary of an error in $\Gamma^\ast$ to be the boundary of the volume that corresponds to the collection of  qubits on which the error acts nontrivially. In other words,
\begin{eqnarray}
\partial E = \sum_{\nu:E_\nu\neq I }\partial\nu =  \partial (\supp(E)).\label{eq:error-bndry}
\end{eqnarray}
We use the same notation $\partial $ for boundary of cells as well as operators. Note that $\partial (EE') = \partial E +\partial E'$.

\begin{lemma}[$X$ error boundary]\label{lm:x-err-est}
Using the same notation as in  Theorem~\ref{lm:x-err-syn},
the boundary of an error $E$ in $\Gamma^\ast$ is 
\begin{eqnarray}
\partial E = \sum_{c}  \supp(\pi_{c}(E)).
\end{eqnarray}
\end{lemma}

\begin{proof}
Let $E_\nu$ denote the error on the 3-cell corresponding to the $\nu$th qubit.  Then  we can write 
\begin{eqnarray}
\partial E & \overset{(a)}{=} & 
\sum_{\nu: E_\nu\neq I } \partial \nu 
\overset{(b)}{=} \sum_{\nu: E_\nu \neq I}\sum_{c} \pi_c(\nu)\\
& =& \sum_{c} \sum_{\nu: E_\nu \neq I } \pi_c(\nu)\overset{(c)}{=} \sum_{c}\sum_{\nu:E_\nu\neq I}\supp(\pi_c(E_\nu))\\
&\overset{(d)}{=}&\sum_c \supp(\pi_c(E)) 
\end{eqnarray}
where $(a)$ follows from the definition of the boundary of an error; (b) follows from  
the fact that the boundary of a single qubit is the collection of four faces of the 
tetrahedron that  correspond to the qubit; 
$(c)$ follows from  rearranging the order of summation and the observation that $\pi_c(\nu)$ is the $c$-face in the boundary of $\nu$, see Eq.~\eqref{eq:cell-bndry},  and if $E_\nu\neq I$, then this is the same as the support of $\pi_c(E_\nu)$; $(d)$ follows from Eq.~\eqref{eq:supp-sum} and completes the proof. 
\end{proof}

By Lemma~\ref{lm:x-err-est}, the boundary of an error $E$ can be broken down into four surfaces each lying in a separate minor complex $\Gamma^{\ast\setminus c}$. We show that if these surfaces were modified by the support of a stabilizer in the minor complexes, then these modified surfaces form the boundary of an error $E'$ which is equivalent to $E$ up to a stabilizer i.e. 
$E'=ES$ for some $X$ stabilizer $S$ on the color code.

\begin{lemma}[$X$ error boundary modulo stabilizer]\label{lm:lifting-modulo} 
Let $E$ be an error on $\Gamma^\ast$. Let $S_c$ be $X$-stabilizer generators on $\Gamma^{\ast\setminus c}$. Then 
$\partial E + \sum_c \supp(S_c)$ is the boundary of $ES$, for some  $X$-stabilizer $S$ on $\Gamma^\ast$ i.e. 
\begin{eqnarray}
\partial E + \sum_c \supp(S_c) = \partial (ES)
\end{eqnarray}
\end{lemma}
\begin{proof}
 Since $S_c$ is a stabilizer generator, by Corollary~\ref{co:stab2stab}, there exists an $X$-stabilizer $\overline{S}_c$ on 
 $\Gamma^\ast$ such that  $\pi_c(\overline{S}_c)=S_c$ and $\pi_{c'}(\overline{S}_c)=I$.  By Lemma~\ref{lm:x-err-est}, 
 $\partial(\overline{S}_c) = \supp(S_c) $. 
  Then using 
 $\partial EF = \partial E +\partial F$, we obtain  
$ \partial E\overline{S}_c = \partial E + \partial \overline{S}_c$. Repeating this for all $c$ we have
$ \partial (E\prod_c\overline{S}_c) = \partial E + \sum_c \partial \overline{S}_c$. Letting $\prod_c \overline{S}_c=S$, we
 can write this as  $\partial (E S) = \partial E + \sum_c \supp(S_c)$
as claimed. 
\end{proof}
The importance of the previous result is that we can independently estimate the four components of boundary of an error.  
With these results in hand we can  estimate the boundary of an $X$ error.

\begin{theorem}[Estimating face boundary of $X$ errors]\label{th:x-errors}
The boundary an X-error $E$ on the dual of color code can be estimated by Algorithm~\ref{alg:x-type}.
The algorithm estimates $\partial E$ up to the boundary of an $X$ stabilizer of the color code, provided $\pi_c(E)$ is estimated up 
to a stabilizer on 
$\Gamma^{\ast\setminus c}$, where  $c\in\{r,g,b,y \}$. 
\end{theorem}
\begin{proof} 
We only  sketch the proof as it is a straightforward consequence of the results we have shown thus far. 
By Lemma~\ref{lm:x-err-est} the boundary of the error consists of support of $\pi_c(E)$ i.e. the projections of the error on the 3D toric codes on $\Gamma^{\ast\setminus c}$. By Lemma~\ref{lm:x-err-syn}, the syndrome of $\pi_c(E)$ is the restriction of the syndrome on $\Gamma^\ast$. Therefore, $\pi_c(E)$ can be estimated by decoding on $\Gamma^{\ast\setminus c}$. By Lemma~\ref{lm:lifting-modulo}, if the estimate for $\pi_c(E)$ is equivalent up to a stabilizer on 
$\Gamma^{\ast\setminus c}$, we can obtain the boundary of $E$ up to the boundary of a stabilizer on the color code. 
\end{proof}

\begin{algorithm}[H]
\caption{{\ensuremath{\mbox{ Estimating (face) boundary of $X$ type error}}}}\label{alg:x-type}
\begin{algorithmic}[1]
\REQUIRE {A 3-colex $\Gamma$, Syndrome of an $X$ error $E$ }
\ENSURE {  $\mathsf{F}$,  an estimate of $\partial E$ where $\mathsf{F}\subseteq \mathsf{C}_2(\Gamma^\ast)$}
\FOR {each $c \in \{r,b,g,y\} $}
\FOR { each edge $e$ in $\Gamma^{* \setminus c}$ } \hfill // syndrome projection
\STATE $s_{\pi_c(e)} = s_e$ \hfill // $s_e$ is syndrome on edge $e$
\ENDFOR
\STATE Using the projected syndrome on $\Gamma^{* \setminus c}$, estimate the error $\mathsf{F}_c$ by any 3D toric code decoder  for $X$ errors  

\ENDFOR
\STATE Return $\mathsf{F}=\bigcup_c \mathsf{F}_c$
\end{algorithmic}
\end{algorithm}
If any of  component decoders on the minor complexes make a logical error, then Algorithm~\ref{alg:x-type} may fail to produce a valid boundary. 

\subsection{$Z$ type errors on 3D color codes }
One can prove  results similar to previous section for the $Z$-type errors also. As we noted earlier, the syndrome information about $Z$ errors resides on the vertices of $\Gamma^\ast$ while the error corresponds to a volume. To recover the boundary of the error,
the objects of interest are the minor complexes $\Gamma^{\ast\setminus cc'}$. Recall that the edges of $\Gamma^{\ast\setminus cc'}$
are associated with qubits. So we project the errors on cell $\nu$ to $\pi_{cc'}(\nu)$ which is an edge in $\Gamma^{\ast\setminus cc'}$.
We define 
\begin{eqnarray}
\pi_{cc'}(Z_\nu) = Z_{\pi_{cc'}(\nu)}
\end{eqnarray}

Let syndrome on $v\in \mathsf{C}_0(\Gamma^\ast)$ be $s_v$. Then we define the syndrome on $v\in \Gamma^{\ast\setminus cc'}$ as 
\begin{eqnarray}
\pi_{cc'}(s_v) = s_{\pi_{cc'}(v)}=s_v
\end{eqnarray}
so the syndrome on $\pi_{cc'}(\Gamma^\ast)$ is simply the restriction of the syndrome on $\Gamma^\ast$. %For us to be able 
To project the color code onto toric codes, this syndrome must be a valid syndrome on the minor complex.

We next show that both $Z$ errors and their associated syndromes can be projected consistently onto the minor complexes 
$\Gamma^{\ast\setminus cc'}$.

\begin{theorem}[Projection of $Z$ errors onto toric codes]\label{lm:z-err-syn}
Let $E$ be a $Z$-type error on $\Gamma^\ast$ and its  associated syndrome $s$. 
Then the error  $\pi_{cc'}(E)$ in $\Gamma^{\ast\setminus cc'}$ produces the syndrome $\pi_{cc'}(s)$.
\end{theorem}
\begin{proof}
We only need to show the theorem for vertices   $ v \in \Gamma^{\ast\setminus cc'}$. 
Suppose $E$ produces the syndrome $s_v$
on $v$, then 
$s_v = \bigoplus_{\nu :v \in  \nu} q_\nu$, 
where $q_\nu=1$, if there is a $Z$ error on $\nu$ and zero otherwise.
We need to show that $\pi_{cc'}(E)$ produces the syndrome $s_v$ on $v $.
 Every qubit incident on $v$ must have a $dd'$-edge incident on 
$v$, since $v$ must be a $d$ or $d'$-vertex. Two qubits incident on $v$ can share at most one such edge. 
Then we can partition the qubits incident on $v$ depending on the $dd'$ edge on which they are incident.
Let $\{e_1, \ldots, e_m\}$ be these $dd'$-edges.  Then we can write 
\begin{eqnarray}
\{ \nu :v\in \nu\} &=& \cup_{i} \{\nu : v,e_i\in \nu\}\\
s_v &=& \bigoplus_{\nu : v\in \nu} q_\nu = \bigoplus_{i} \bigoplus_{\nu:v,e_i \in \nu } q_\nu = \bigoplus_{i} r_i,
\end{eqnarray}
where $r_i= \bigoplus_{i:v,  e_i\in \nu_i } q_\nu $. 

The qubits containing $e_i$ are projected onto the $dd'$-edge $e_i$ in the minor complex $\Gamma^{\ast\setminus cc'}$ and there is an error on $e_i$ if and only if $r_i=\bigoplus_{\nu:v,  e_i\in \nu } q_\nu =1$. 
Thus the syndrome on the vertex $v$ as computed with respect to the toric code on $\Gamma^{\ast\setminus cc'}$ is $\pi_{cc'}(E)=\bigoplus_{i=1}^m r_i =s_v=\pi_{cc'}(s_v)$ as required. 
\end{proof}

\begin{corollary}[Validity of $Z$ syndrome restriction]\label{co:z-syndrome}
Let $s$ be the syndrome  for a $Z$ error on $\Gamma^\ast$. Then  $\pi_{cc'}(s)$ is a valid syndrome for a 
$Z$ error on $\Gamma^{\ast\setminus cc'}$.
\end{corollary}
\begin{proof}
From Theorem~\ref{lm:z-err-syn}, we see that $\pi_{cc'}(s)$ coincides with the syndrome produced by a $Z$ error on $\Gamma^{\ast\setminus cc'}$. 
Hence, $\pi_{cc'}(s)$ must be a valid syndrome for a $Z$ error on a 3D toric code. 
\end{proof}

Having projected both the error and the syndrome onto the minor complexes, we recover the boundary of the error in steps. To this end we  define the edge boundary of an error $E$ as 
\begin{eqnarray}
\delta E &=& \sum_{\nu:E_\nu\neq I }\delta\nu = \sum_{\nu:E_\nu\neq I }\sum _{cc'} \pi_{cc'}(\nu)\label{eq:edge-bndry-error}\\
&=&\sum_{cc'}{\supp(\pi_{cc'}(E))},\label{eq:edge-bndry-error-a}
\end{eqnarray}
where Eq.~\eqref{eq:edge-bndry-error-a}  follows from  interchanging the order of summation in Eq.~\eqref{eq:edge-bndry-error}.
We see that the edge boundary of $E$ can be recovered by recovering $\pi_{cc'}(E)$. 
The next lemma shows an interesting property of the edge boundary that will help us in correcting $Z$ errors. 

\begin{lemma}[Edge boundary corresponds to an $X$-syndrome] 
\label{lm:z2xerror}
Let $E_Z =\prod_{\nu \in \Omega} Z_\nu$ and $E_X=\prod_{\nu\in \Omega}X_\nu$. Then the syndrome of 
$E_X$ is nonzero on the edges in $\delta E_Z$.
\end{lemma}
\begin{proof}
Consider the syndrome of $E_X$ on an edge $e$. 
Then 
\begin{eqnarray}
s_e\neq0 \mbox{ iff. } |\{ \nu \in \Omega \mid e\in \nu \} | \mbox { is odd}\label{eq:nonzero-syndrome}
\end{eqnarray}
In other words,  $s_e$ is nonzero if and only if the number of qubits in $\Omega$ incident on $e$ is odd. 
From Eq.~\eqref{eq:edge-bndry-error}, we see that  $e$ will be  present in the edge boundary of $E_Z$ if and only if an odd number of qubits are incident on 
$e$. Thus the syndrome of $E_X$ is nonzero on the edge boundary of $E_Z$.
\end{proof}

We cannot always expect to estimate the edge boundary exactly because the estimates for $\pi_{cc'}(E)$ on the minor complexes 
could be off. The next lemmas show this will not be problem as long as the estimates for $\pi_{cc'}(E)$ are off by stabilizer elements. 

{
\begin{lemma}\label{lm:z-est-stab}
Suppose $S$ is a $Z$-stabilizer on $\Gamma^{\ast\setminus cc'}$, then there exists a $Z$-stabilizer $\overline{S}$ in $\Gamma^\ast$ such that
$\pi_{cc'}(\overline{S}) = S$ and $\pi_{cd}(\overline{S}) = I$ for $xy\neq cc'$.
\end{lemma}
\begin{proof}

A $Z$-stabilizer in $\Gamma^{\ast\setminus cc'}$ generated by the face type stabilizers in 
$\Gamma^{\ast\setminus cc'}$. Therefore, it suffices to consider when $S$ is a face type stabilizer. 
By Lemma~\ref{lm:faces-in-cc}, the faces of $\Gamma^{\ast\setminus cc'}$ are in correspondence with the $cc'$-edges of $\Gamma^\ast$. So we can let $S=B_{f_e}^Z$ for some face  $f_e$ in $\Gamma^{\ast\setminus cc'}$
and $cc'$-edge $e$ in $\Gamma^\ast$. The stabilizer of the color code attached to $e$ is given by 
$B_e^Z = \prod_{\nu:e\in \nu } Z_\nu$. Then  
\begin{eqnarray}
\pi_{cc'}(B_e^Z) &= &\prod_{\nu:e\in \nu } Z_{\pi_{cc'}(\nu)}  \stackrel{(a)}{=} \prod_{\nu:e\in \nu}Z_{e_{\ni\nu}^{dd'}},\\
&\stackrel{(b)}{=}& \prod_{t\in \partial(f_e)} Z_t = B_{f_e}^Z  = S,
\end{eqnarray}
where $(a)$ follows from Eq.~\eqref{eq:pi-cc-cell} and $(b)$ from Eq.~\eqref{eq:bndry-f-cc}.

Denote by $\mathcal{V}_e$ the set of qubits incident on $e$. Every qubit $\nu \in \mathcal{V}_e$ contains a $c'd'$ edge. These edges must also be incident on the $c'$-vertex of $e$. Two qubits $\nu$ and $\nu'$ which have the same $c'd'$
edge must share a face since they already share $e$. Hence only two qubits $\nu, \nu' \in \mathcal{V}_e$ can share a 
$c'd'$-edge.  For these qubits we have $\pi_{cd}(\nu) =\pi_{cd}(\nu')$. This implies 
$\pi_{cd}(B_e^Z) = \prod_{\nu:e\in \nu}Z_{e_{\ni\nu}^{c'd'}} =I$. 
Similar arguments can be used to show that $\pi_{cd}(\overline{S}) = I$ for other $xy\neq cc'$. We omit the details. 
\end{proof}
}

\begin{lemma}[Edge boundary modulo stabilizer]\label{lm:z2xsyndrome}
Let $E$ be a $Z$-type error on the color code and $S$  a $Z$-stabilizer on $\Gamma^{\ast\setminus cc'}$. 
Then $\delta E+ \sum_{e:S_e\neq I}e$ 
is the edge boundary of $E\overline{S}$ for some $Z$-stabilizer  $\overline{S}$ on $\Gamma^\ast$.
\end{lemma}
\begin{proof}
By Lemma~\ref{lm:z-est-stab}, there exists a $Z$ stabilizer  ${\overline{S}}$ on $\Gamma^\ast$ such that 
$\pi_{cc'}({\overline{S}})=S$ and $\pi_{xy}({\overline{S}})=I$ for $xy\neq cc'$. 
Therefore, 
\begin{eqnarray}
I &=&\pi_{xy}(\overline{S})=  \pi_{xy}\left(\prod_{\nu:\overline{S}_\nu\neq I} Z_\nu \right) = \prod_{\nu:\overline{S}_\nu\neq I} Z_{\pi_{xy}(\nu)}
\end{eqnarray}
From this we infer that for $xy\neq cc'$, 
\begin{eqnarray}
\sum_{\nu:\overline{S}_\nu\neq I } \pi_{xy}(\nu)&=&0
\end{eqnarray}
Therefore, the edge boundary of $\overline{S}$  has support only in $\Gamma^{\ast\setminus cc'}$. 
Furthermore, substituting for $S$ and $\overline{S}$ in $S=\pi_{cc'}(\overline{S})$ we obtain 
\begin{eqnarray}
\prod_{e:S_e\neq I }Z_e&=& \pi_{cc'}\left(\prod_{\nu:\overline{S}_\nu\neq I} Z_\nu \right)=\prod_{\nu:\overline{S}_\nu\neq I} Z_{\pi_{cc'}(\nu)} \label{eq:bndry-1}
\end{eqnarray}
Eq.~\eqref{eq:bndry-1}, implies that 
\begin{eqnarray}
\sum_{e:S_e\neq I } e &= &\sum_{\nu:\overline{S}_\nu\neq I } \pi_{cc'}(\nu)\\
&=& \sum_{\nu:\overline{S}_\nu\neq I } \pi_{cc'}(\nu)+\sum_{xy\neq cc'}\sum_{\nu:\overline{S}_\nu\neq I } \pi_{xy}(\nu)\\
&=&\delta \overline{S}
\end{eqnarray}
Thus $\delta \overline{S} = \sum_{e:S_e\neq I }e$ and  $\delta E+ \sum_{e:S_e\neq I}e =\delta(E\overline{S})$. 
\end{proof}
With these results in hand we show how to estimate the edge boundary of a $Z$-error from the minor complexes given the syndrome on its vertices. 

\begin{theorem}[Estimating edge boundary of $Z$ errors]\label{th:z-edge-bndry}
Let $\Gamma$ be a 3-colex and $E$ a $Z$-error on the associated color code. 
Algorithm~\ref{alg:z-type} estimates $\delta E$, the edge boundary of $E$,  up to the boundary of a $Z$ stabilizer of the color code, provided $\pi_{cc'}(E)$ is estimated up 
to a stabilizer on  $\Gamma^{\ast\setminus cc'}$.
\end{theorem}
\begin{proof}
By Corollary~\ref{co:z-syndrome}, the restriction of the syndrome of $E$ is a valid syndrome on $\Gamma^{\ast\setminus cc'}$. By 
Theorem~\ref{lm:z-err-syn}, $\pi_{cc'}(E)$ has the same syndrome as the restriction and $\pi_{cc'}(E)$ can be estimated using a 3D toric decoder on 
$\Gamma^{\ast\setminus cc'}$.
We can reconstruct the edge boundary from $\pi_{cc'}(E)$ using the Eq.~\eqref{eq:edge-bndry-error-a}. By 
 Lemma~\ref{lm:z2xsyndrome} if the estimates for $\pi_{cc'}(E)$ are upto a stabilizer on $\Gamma^{\ast\setminus cc'}$, the estimate  for edge boundary of $E$ will differ by the  boundary of a $Z$-stabilizer on the color code. 
\end{proof}

\begin{algorithm}[H]
\caption{{\ensuremath{\mbox{ Estimating (edge) boundary of $Z$ type error}}}}\label{alg:z-type}
\begin{algorithmic}[1]
\REQUIRE {A 3-colex $\Gamma$,  Syndrome of a $Z$ error $E$ }
\ENSURE {$\mathsf{E}$,  an estimate of the edge boundary $\delta E$,  where $\mathsf{E}\subseteq \mathsf{C}_1(\Gamma^\ast)$}
\FOR {each $c,c' \in \{r,b,g,y\} $}
\FOR { each vertex $v$ in $\Gamma^{*\setminus cc'}$ } \hfill //syndrome projection
\STATE $s_{\pi_{cc'}(v)} = s_v$ \hfill  //$s_v$ is  syndrome on vertex $v$
\ENDFOR

\STATE Estimate the error $\mathsf{E}_{cc'}$ using any 3D toric code decoder for $Z$ errors on $\Gamma^{* \setminus cc'}$
\ENDFOR
\STATE Return $\mathsf{E}=\sum_{c,c'} \mathsf{E}_{cc'}$

\end{algorithmic}
\end{algorithm}

What we have achieved so far  is that we have taken a $Z$ error whose syndrome is on vertices and converted it to a valid syndrome on edges for an $X$ error with the same support. 
We can take this syndrome on edges and recover the face boundary of the error using 
Theorem~\ref{th:x-errors}. 

Note that the  estimate for edge boundary returned by Algorithm~\ref{alg:z-type} need not be a valid edge boundary if
any of the component decoders fail. So we need a method to check the validity of the edge boundary. Recall that the $X$-syndrome  on the toric code is boundary of a collection of faces. Therefore it is a union of homologically trivial cycles. 
We can project the edge boundary $\mathsf{E}$ obtained in Algorithm~\ref{alg:z-type} onto each of the minor complexes $\Gamma^{\ast\setminus c}$. If the edge boundary is valid, then all the homologically nontrivial closed surfaces
 in $(\Gamma^{\ast\setminus c})^\ast$ will intersect with the projected syndrome an even number of times. This test can be carried out in linear time in the number of qubits.

\begin{theorem}[Estimating face boundary of $Z$ type errors] \label{th:z-errors} 
Let $E$ be a $Z$-type error on $\Gamma^\ast$ whose edge boundary is estimated using Algorithm~\ref{alg:z-type}. If the edge boundary is valid, then we can estimate the face boundary of 
$E$ up to a $Z$-stabilizer on the color code using Algorithm~\ref{alg:x-type}. 
\end{theorem}
\begin{proof}
By Theorem~\ref{th:z-edge-bndry} we can estimate the edge boundary of $E$  up to a $Z$-stabilizer boundary. By Lemma~\ref{lm:z2xsyndrome} these edges are precisely the syndrome for an $X$ error. By Lemma~\ref{lm:z2xerror}, this $X$ error has the same support as $E$ up to a $X$-stabilizer. But in a color code for every $X$-stabilizer there exists a $Z$-stabilizer with the same support. 
Thus the final boundary of $E$ is estimated up to a $Z$-stabilizer provided all the intermediate estimates from 
Algorithms~\ref{alg:x-type}~\&~\ref{alg:z-type} 
are all up to a stabilizer on the respective 3D toric code decoders. 
\end{proof}

\subsection{ Decoding 3D color codes}\label{ssec:decoding}
The projection onto the toric codes allow us to decode the 3D color code. Before we can give the complete decoding algorithm we need one more component. 
Following Theorems~\ref{th:x-errors}~and~\ref{th:z-errors}, we only end up with the boundary of the error. 
We need to  identify the qubits which are in error. The procedure for  lifting the boundary to volume is given in
Algorithm~\ref{alg:lifting}. The main idea behind this algorithm is the fact that the color code is connected and we can partition the qubits into two groups: those inside and those outside of the boundary. The following lemma justifies the procedure in Algorithm~\ref{alg:lifting}. 

\begin{algorithm}[H]
\caption{{\ensuremath{\mbox{Lifting a boundary to a volume}}}}\label{alg:lifting}
\begin{algorithmic}[1]
\REQUIRE {Complex $\Gamma^*$, Set of faces $\mathsf{F} \subseteq \mathsf{C}_2(\Gamma^\ast)$}
\ENSURE {$\Omega \subseteq \mathsf{C}_3(\Gamma^\ast)$ such that $\mathsf{F}$ is the boundary of $\Omega$}

\STATE Set $\Omega= \emptyset;  m_{\nu} =0$ for all $ \nu \in  \mathsf{C}_3(\Gamma^\ast)$ \hfill // Initialization
\STATE Initialize $\Omega =\{\nu_o \}, m_{\nu_o} =1$   \hfill // For some 3-cell $\nu_o$
\WHILE{$m_\mu =0  $ for some $\mu\in \mathcal{N}_\nu$ with $m_\nu\neq 0$} 
\FOR{each $\mu \in \mathcal{N}_\nu$} \hfill // $\mathcal{N}_v:=$ 3-cells sharing a face with $\nu$
\IF{$m_\mu = 0$}
\IF {$\mu\cap \nu \in \mathsf{F}$}
\STATE $m_\mu= -m_\nu$ \hfill //Qubits on different side of error boundary
\ELSE 
\STATE $m_\mu= m_\nu$ \hfill //Qubits on same side of error boundary
\ENDIF
\IF{$m_\mu = 1$}
\STATE $\Omega = \Omega \cup \{\mu \}$
\ENDIF

\ELSE 
\IF {$\mu\cap \nu \in \mathsf{F}$ and $m_\mu\neq -m_\nu$ }
\STATE $\Omega =\emptyset$; Exit \hfill// $\mathsf{F}$ not a valid boundary 
\ENDIF
\IF {  $\mu\cap \nu \not\in \mathsf{F}$ and $m_\mu\neq m_\nu$}
\STATE $\Omega =\emptyset$; Exit \hfill// $\mathsf{F}$ not a valid boundary 
\ENDIF

\ENDIF
\ENDFOR
\ENDWHILE

\IF{$|\Omega| > 
|\mathsf{C}_3\setminus\Omega|$}
\STATE $\Omega =\mathsf{C}_3\setminus\Omega$ \hfill //Pick the smaller volume
\ENDIF
\end{algorithmic}
\end{algorithm}

\begin{lemma}[Lifting the boundary of error]\label{lm:lifting}
Algorithm~\ref{alg:lifting} will give the smallest collection of $3$-cells $\Omega \subseteq \mathsf{C}_3(\Gamma^\ast)$ 
such that $ \partial \Omega  = \mathsf{F}$.  If $\mathsf{F}$ is not a valid boundary, then 
the algorithm returns an empty set. 
\end{lemma}
\begin{proof}
The algorithm takes as input a collection of faces supposed to enclose a volume. If the faces enclose a volume, we can label all the 3-cells inside and outside the boundary  differently. Cells adjacent to each other and enclosed within the same boundary are labeled same. 
The algorithms proceeds by labeling a random choice of initial qubit and then proceeds to assign labels to all its adjacent qubits. If two qubits share a face that is not in the boundary $\mathsf{F}$, then they must have the same label because one must cross the boundary to change the label. Two qubits, that are adjacent and share a face that is in the boundary must have different labels. 
The algorithm stops when there are no more qubits to be labeled or when a qubit  is assigned contradicting labels, indicating that $\mathsf{F}$ is not a boundary. 
\end{proof}
The running time of the algorithm is linear in the number of qubits. %$O(|\mathsf{C}_3|)$. 
The algorithm assumes that all qubits have the same error probability. It 
can be modified  so that it picks the most likely qubits if the error probabilities are not uniform.
We now give the decoding procedure for color codes.  
\begin{theorem}[Decoding 3D color codes via 3D toric codes]\label{th:decoding-tcc}
An error $E$ on a color code can be estimated using Algorithm~\ref{alg:css-decoding}. The estimate will be within a stabilizer on the color code provided the intermediate decoders also estimate within a stabilizer on the respective codes. 
\end{theorem}
\begin{proof}
The proof of this theorem is straightforward given our previous results. The decoding is performed separately for $X$ and $Z$ errors
and it makes use of the fact that the color code is a CSS code. 
The algorithm proceeds by  estimating the boundary of the 
$X$-type errors and $Z$-errors separately.  The correctness of these procedures is due to Theorems~\ref{th:x-errors}~and~\ref{th:z-errors}. Lemma~\ref{lm:lifting} ensures that these boundaries can be lifted to find the qubits that are in error. Decoding failure results if any of the component decoders fail or make logical errors. This will lead to either the failure of lifting procedure or an invalid edge boundary in line 7. The validity of the edge boundary can be checked by 
ensuring that the restricted syndrome $\pi_c(\delta E_Z)$ consists of homologically trivial cycles. 
\end{proof}
\begin{algorithm}[H]
\caption{{\ensuremath{\mbox{ Decoding 3D color codes}}}}\label{alg:css-decoding}
\begin{algorithmic}[1]
\REQUIRE {A 3-colex $\Gamma$ and the syndrome}
\ENSURE {Error estimate $\hat{E} $ }
\STATE Let $s_X$ be syndrome for $X$ type error $E_X$
\STATE Obtain the face boundary $\partial E_X$ from  Algorithm~\ref{alg:x-type} with $s_X$ as input 
\STATE Lift the boundary $\partial E_X$ by running Algorithm~\ref{alg:lifting} and obtain  $\Omega_X$, the support of $E_X$
\IF  {$\Omega_X=\emptyset $ and $s_X\neq 0$ } 
\STATE Declare decoder failure and exit                                                                                                                                                                                                                                                                                                                         
 \ENDIF
\STATE Let $s_Z$ be syndrome for $Z$ type error $E_Z$
\STATE Estimate  the edge boundary $ \delta E_Z$ from  Algorithm~\ref{alg:z-type} with $s_Z$ as input 
\STATE Check $\pi_c(\delta E_Z)$ consists of homologically trivial cycles only, otherwise declare decoding failure and exit. 
\STATE Obtain the face boundary $\partial E_Z$ from  Algorithm~\ref{alg:x-type} with $\delta E_Z$ as input
\STATE Lift the boundary $\partial E_Z$ by running Algorithm~\ref{alg:lifting} and obtain  $\Omega_Z$, the support of $E_Z$
\IF {$\Omega_Z=\emptyset $ and $s_Z\neq 0$}
\STATE Declare decoder failure and exit                                                                                                                                                                                                                                                                                                                         
\ENDIF
\STATE  Return $\hat{E} = \prod_{\nu\in \Omega_X} X_{\nu} \prod_{\nu\in \Omega_Z}Z_\nu$
\end{algorithmic}
\end{algorithm}
\begin{remark}
The decoder could fail if any of the intermediate decoders make a logical error. 
\end{remark}
The overall running time depends on the running time of the 3D toric code decoders. We can  run them independently or we can take advantage of the fact that the errors on component 3D toric codes are correlated. 

\section{Conclusion}

In this paper we have shown how to project 3D color codes onto 3D toric codes. The projection was motivated by the problem of decoding 
3D color codes. The toric codes  thus obtained are linearly related to the size of the parent color code. So if 3D toric codes on arbitrary lattices can be decoded efficiently, then so can the 3D color codes by projecting them onto 3D toric codes using our map. Our work provides an alternative perspective to that of \cite{kubica15} who also proposed a map between color codes and toric codes. Our approach emphasizes the topological properties of color codes. 
At this point there is no data available for performance of the decoders based on our map as well as the map due to \cite{kubica15}.
One difficulty to compare the performance of the decoders arising out of these maps, as we mentioned earlier,  is that we do not have efficient decoders for 3D toric code on an arbitrary lattice. 
(Decoders are known only for the cubic lattice.)
So an open question for further research is to study the decoding of 3D toric codes. 
This map could also find application in the decoding of gauge color codes \cite{brown15,bombin15,bombin15a} by projecting onto 3D toric codes. 
 Another avenue for further research is to study the possible use of this map for fault tolerant quantum computing protocols. 
\subparagraph*{Acknowledgement}
This research was supported by a grant from Center for Industrial Consultancy and Sponsored Research.

% \bibliography{refer}
 
 %merlin.mbs apsrev4-1.bst 2010-07-25 4.21a (PWD, AO, DPC) hacked
%Control: key (0)
%Control: author (8) initials jnrlst
%Control: editor formatted (1) identically to author
%Control: production of article title (-1) disabled
%Control: page (0) single
%Control: year (1) truncated
%Control: production of eprint (0) enabled
%

\end{document}